\documentclass[letterpaper, 10 pt, conference]{ieeeconf}  

\IEEEoverridecommandlockouts                              
\usepackage{amsmath,amssymb, graphicx,epstopdf}
\usepackage{cite}
\usepackage[tight,footnotesize]{subfigure}
\usepackage[usenames,dvipsnames]{color}

\newtheorem{Theorem}{Theorem}
\newtheorem{Lemma}{Lemma}

\newtheorem{Remark}{Remark}

\newtheorem{Condition}{Condition}

\newtheorem{Definition}{Definition}
\newcommand{\IEEEbq}{\begin{IEEEeqnarray}{rCl}}
\newcommand{\IEEEeq}{\end{IEEEeqnarray}}

\def \node {\mathcal N}
\def \generator {\mathcal G}
\def \load {\mathcal L}

\title{\LARGE \bf
Decentralized Primary Frequency Control in Power Networks
}

\author{Changhong Zhao and Steven Low
\thanks{This work was supported by NSF NetSE grant CNS 0911041, ARPA-E grant DE-AR0000226, Southern California Edison, National Science Council of Taiwan R.O.C. grant NSC 103-3113-P-008-001, Caltech Resnick Institute, and 
California Energy Commission's Small Grant Program through Grant 57360A/11-16. }
\thanks{C. Zhao and S. Low are with the Department of Electrical Engineering, California Institute of Technology, Pasadena, CA 91125, USA 
        {\tt\small \{czhao, slow\}@caltech.edu}}%
}%

\begin{document}

\maketitle
\thispagestyle{empty}
\pagestyle{empty}

\begin{abstract}

We augment existing generator-side primary frequency control with load-side control that are local, ubiquitous, and continuous. The mechanisms on both the generator and the load sides are decentralized in that their control decisions are functions of locally measurable frequency deviations. These local algorithms interact over the network through nonlinear power flows. We design the local frequency feedback control so that any equilibrium point of the closed-loop system is the solution to an optimization problem that minimizes the total generation cost and user disutility subject to power balance across entire network. With Lyapunov method we derive a sufficient condition ensuring an equilibrium point of the closed-loop system is asymptotically stable. Simulation demonstrates improvement in both the transient and steady-state performance over the traditional control only on the generators, even when the total control capacity remains the same.

\end{abstract}

\section{INTRODUCTION}

It is important to maintain the frequency of a power system tightly around its nominal value for the quality of power delivery and safety of infrastructures. Frequency is mainly affected by real power imbalance through dynamics of rotating units in the system, and hence frequency is usually controlled by controlling real power injections. Frequency/real power control is traditionally implemented on the generator side and usually consists of different layers that work at different timescales in concert as described in \cite{Ilic2007hierarchical, KianiAnnaswamy2012}. In this paper we focus on the layer of primary frequency control, which operates at the timescale of seconds and provides preliminary power balance to stop frequency excursion following a change in power injection. Generator-side primary frequency control is also known as droop control, in which a speed governor adjusts the setpoint of the turbine to change power generation based on local frequency feedback.

On the other hand, ubiquitous continuous distributed load participation in frequency control has started to play a rising role since the last decade or so. The idea is that household appliances such as refrigerators, heaters, ventilation systems, air conditioners, and plug-in electric vehicles can be controlled to help manage power imbalance and regulate frequency. This idea dates back to the late 1970s \cite{Schweppe1980}, and has been extensively studied recently \cite{Short2007, brooks 2010, CallawayHiskens2011}, with particular focus on the primary timescale to benefit from the fast-acting capability of frequency-responsive loads \cite{donnelly2010frequency, molina2011decentralized}. There are also load-side frequency control field trials carried out around the world \cite{Hammerstrom2007, UKProgram2008}. Simulation-based studies and field trials above have shown significant improvement in performance and reduction in the need for generator reserves.

For both generator and load-side frequency controls, there are two issues: design and stability. In the literature, design of generator frequency control either happens at the secondary layer or above to incorporate economic dispatch without changing the primary layer \cite{Ilic2007hierarchical, KianiAnnaswamy2012}, or modifies the existing primary control for stability \cite{JiangCaiDorseyQu1997, Siljak2002}. In designing load-side frequency control, most of the existing works do not take into account the disutility to users for participating in control. Previous studies, e.g., \cite{BergenHill1981, tsolas1985structure}, on stability of power networks with linear frequency dependent loads can be generalized to the case with load-side frequency feedback control, but most of them ignore the dynamics of governors and turbines in generator frequency control loops. Moreover, all the studies above consider generator-side and load-side controls separately, without investigating the interactions between them.    

In this paper we provide a systematic method to design generator and load-side primary frequency control in a unified model. The control is completely decentralized in that every generator and load makes its decision based on local frequency sensing and its own cost of generation or consumption. However, the control attains equilibrium points which maximize the benefit across the entire network as defined by an optimization problem called \emph{optimal frequency control} (OFC). We apply Lyapunov method to derive a sufficient condition for the equilibrium points of the closed-loop system to be asymptotically stable. Simulation with a more realistic model shows that the proposed control improves both the transient and steady-state frequency, compared with traditional control on generators only. Interestingly this is achieved with the same total control capacity.

 In our previous work \cite{ZhaoTopcuLiLow2014} we designed a similar load control and proved its stability for a simplified model with linearized power flows and without generator-side control. This paper extends \cite{ZhaoTopcuLiLow2014} with nonlinear power flows and by adding governor and turbine dynamics to the classical power network model.  

The rest of this paper is organized as follows. Section \ref{sec:model} describes a dynamic model of power networks. Section \ref{sec:design} formulates the OFC problem which characterizes the desired equilibrium and guides the design of decentralized primary frequency control. Section \ref{sec:stability} derives a sufficient condition for the equilibrium points of the closed-loop system to be asymptotically stable. Section \ref{sec:simulation} is a simulation-based case study to show the performance of the proposed control. 
Section \ref{sec:conclusion} concludes the paper and discusses future work. 




\section{POWER NETWORK MODEL}
\label{sec:model}

Let $\mathbb{R}$ denote the set of real numbers. For a set $\node$, let $|\node|$ denote its cardinality. A variable without a subscript usually denotes a vector with appropriate components, e.g., $\omega= (\omega_j, j \in \node)\in  \mathbb{R}^{|\node|}$. For $a,b \in \mathbb{R}$, $a \leq b$, the expression $[\cdot]^b_a$ denotes $\max\left\{\min\{\cdot,~b\},~a\right\}$. 
For a matrix $A$, let $A^T$ denote its transpose. For a square matrix $A$, the expression $A>0$ ($A<0$) means it is positive (negative) definite. For a signal $\omega(t)$ of time $t$, let $\dot \omega$ denote its time derivative ${d \omega}/{d t}$. 

We combine the classical structure preserving model for power network in \cite{BergenHill1981} and the generator speed governor and turbine models in \cite{BergenVittal2000}, \cite{JiangCaiDorseyQu1997}. The power network is modeled as a graph $(\node,\mathcal{E})$ where $\node=\{1,\dots,|\node|\}$ is the set of buses and $\mathcal{E} \subseteq \node \times \node$ is the set of lines connecting those buses. We use $(i,j)$ to denote the line connecting buses $i,j\in \node$, and assume that $(\node,\mathcal{E})$ is directed, with an arbitrary orientation, so that if $(i,j)\in \mathcal E$ then $(j, i)\not\in \mathcal E$. We also use ``$i: i\rightarrow j$'' and ``$k: j\rightarrow k$'' respectively to denote the set of buses $i$ that are 
predecessors of bus $j$ and the set of buses $k$ that are successors of bus $j$. We assume without loss of generality that
 $(\node, \mathcal{E})$ is connected, and make the following assumptions which have been well-justified for transmission networks:
\begin{itemize}
\item The lines $(i,j) \in \mathcal E$ are lossless and characterized by their reactances $ x_{ij}$.
\item The voltage magnitudes $|V_j|$ of buses $j \in \node$ are constants.
\item Reactive power injections on buses and reactive power flows on lines are not considered.
\end{itemize}

A subset $\generator \in \node$ of the buses are fictitious buses representing the internal of generators. Hence we call $\generator$ the set of generators and $\load := \node\setminus\generator$ the set of load buses.\footnote{We call all the buses in $\load$ load buses without distinguishing between generator buses (buses connected directly to generators) and actual load buses,  since they are treated in the same way mathematically.} We label the buses so that $\generator=\{1,\dots,|\generator|\}$ and $\load=\{|\generator|+1,\dots,|\node|\}$.

The voltage phase angle of bus $j \in \node$, with respect to the rotating framework of nominal frequency $\omega_0 = 120 \pi \text{~rad/s}$, is denoted $\theta_j$. Let $\omega_j$ denote the frequency deviation of bus $j$ from the nominal frequency, and we have, for $j \in \node$, that 
\IEEEbq\label{eq:model:freq}
\dot \theta_j = \omega_j.
\IEEEeq 
The frequency dynamics on buses $j \in \node$ are described by the swing equation
\IEEEbq\label{eq:model:swing}
M_j \dot \omega_j  = - D_j \omega_j  + p_j - F_j(\theta)
\IEEEeq
where $M_j>0$ for $j \in \generator$ are parameters characterizing the inertia of rotors in generators and $M_j = 0$ for $j \in \load$. The terms $D_j \omega_j$ with parameters $D_j>0$ for all $j \in \node$ represent (for $j \in \generator$) the deviations in power dissipations due to friction, or (for $j \in \load$) the deviations in frequency-dependent loads, e.g., induction motors. Below we call $D_j \omega_j$ the power dissipation without distinguishing between these two cases. The variable $p_j$ denotes the real power injection to bus $j$, which is the mechanic power injection to the generator for $j\in \generator$, and is the negative of real power load for $j\in \load$. Moreover, 
\IEEEbq\nonumber
F_j(\theta) &:=& \sum\limits_{k:j\rightarrow k} Y_{jk} \sin(\theta_j - \theta_k) \\ 
& & - \sum\limits_{i:i\rightarrow j} Y_{ij} \sin(\theta_i - \theta_j)  \label{eq:model:PF}
\IEEEeq
is the net real power flow out of bus $j$, with parameter $Y_{jk}:=\frac{|V_j||V_k|}{x_{jk}}$ the maximum real power flow on line $(j,k)$. 

Associated with a generator $j \in \generator$ is a system of speed governor and turbine. Their dynamics are described as
\IEEEbq\label{eq:model:governor}
\dot a_j &=& -\frac{1}{\tau_{g,j}} a_j + \frac{1}{\tau_{g,j}}  p_j^c \\
\label{eq:model:turbine}
\dot p_j &=& -\frac{1}{\tau_{b,j}} p_j +  \frac{1}{\tau_{b,j}} a_j
\IEEEeq
where $a_j$ is the position of the valve of the turbine, $p_j^c$ is the control command to the generator, and $p_j$, as introduced above, is the (mechanic) power injection to the generator. The time constants $\tau_{g,j}$ and $\tau_{b,j}$ characterize respectively the time-delay in governor action and the (approximated) fluid dynamics in turbine. There is usually a frequency feedback term $-\frac{1}{R_j} \omega_j$ added to the right-hand-side of \eqref{eq:model:governor}, known as the frequency droop control. Here this term is merged into $p_j^c$ to allow for a general form of frequency feedback control.

Equations \eqref{eq:model:freq}--\eqref{eq:model:turbine} specify a dynamic system with state variables $(\theta, \omega, a_\generator, p_\generator)$ where 
\IEEEbq\nonumber
\theta := \{\theta_1,\dots,\theta_{|\node|}\},&\quad& \omega := \{\omega_1,\dots,\omega_{|\node|}\}  \\ \nonumber
a_\generator := \{a_1,\dots,a_{|\generator|}\}, &\quad& p_\generator := \{p_1,\dots,p_{|\generator|}\}
\IEEEeq
and input variables $(p_\generator^c, p_\load)$ where
\IEEEbq\nonumber
p_\generator^c := \{p_1^c,\dots,p_{|\generator|}^c\},&\quad& p_\load:=\{p_{|\generator+1|},\dots,p_{|\node|}\}. 
\IEEEeq
In the sequel, $(p_\generator^c, p_\load)$ are feedback control to be designed based on the measurement of frequency deviations, i.e., $(p_\generator^c(\omega), p_\load(\omega))$. Parameters $\underline p_j \leq \overline p_j$ specify the bounds of the control variables, e.g., $\underline p_j\leq p_j(\omega) \leq \overline p_j$ or $\underline p_j\leq p_j^c(\omega) \leq \overline p_j$. Note that if $\underline p_j = \overline p_j$ then such bound specifies a constant, uncontrollable input on bus $j$.\footnote{We assume without loss of generality that in OFC below such $j$ are not included in the summation over $\node$ so that Slater's condition is satisfied and strong duality holds.} To motivate the work below, we first define equilibrium points for this closed-loop dynamic system. 

\begin{Definition}\label{def:equilibrium}
An \emph{equilibrium point} of the dynamic system \eqref{eq:model:freq}--\eqref{eq:model:turbine} with feedback control $(p_\generator^c(\omega), p_\load(\omega))$, or a \emph{closed-loop equilibrium} for short, is $(\theta^*, \omega^*, a_\generator^*, p_\generator^*, p_\generator^{c,*}, p_\load^*)$, where $\theta^*$ is a vector function of time and $(\omega^*, a_\generator^*, p_\generator^*, p_\generator^{c,*}, p_\load^*)$ are vectors of real numbers, such that
\begin{IEEEeqnarray}{rClr}
&&p_\generator^{c,*} = p_\generator^c(\omega^*), \quad p_\load^* = p_\load(\omega^*)& \nonumber \\
&&d\theta_j^*/dt = \omega_j^* \quad &  \forall j \in \node  \nonumber \\
 &&- D_j \omega_j^*  + p_j^* - F_j (\theta^*) = 0 \quad&  \forall j \in \node \label{eq:def:eq:3}\\
&&p_j^{c,*} = a_j^* = p_j^* \quad & \forall j \in \generator \label{eq:def:eq:4} \\
&&\omega_i^* = \omega_j^* = \omega^* & \quad \forall i, j \in \node . \footnote{We abuse the notation by using $\omega^*$ to denote both the vector $(\omega_1^*,\dots , \omega_{|\node|}^*)$ and the common value of its components. Its meaning should be clear from the context. }  \label{eq:def:eq:5}
\IEEEeq
\end{Definition}
In the definition above, \eqref{eq:def:eq:3}\eqref{eq:def:eq:4} are obtained by setting right-hand-sides of \eqref{eq:model:swing}\eqref{eq:model:governor}\eqref{eq:model:turbine} to zero, and \eqref{eq:def:eq:5} ensures constant $F(\theta^*)$ at equilibrium points, by \eqref{eq:model:PF}. From \eqref{eq:def:eq:5}, at equilibrium points all the buses are synchronized to the same frequency. The system typically has multiple equilibrium points as will be explained in Section \ref{sec:stability}. We also write the equilibrium point as  $(\theta^*, \omega^*, a_\generator^*,  p_\generator^{c,*}, p^*)$ where $p^*:=(p_\generator^*, p_\load^*)$, when we do not need to distinguish between state variables $p_\generator^*$ and control variables $p_\load^*$.

\section{DECENTRALIZED PRIMARY FREQUENCY CONTROL}\label{sec:design}

An initial point of the dynamic system \eqref{eq:model:freq}--\eqref{eq:model:turbine} corresponds to the state of the system at the time of fault-clearance after a contingency, or the time at which an unscheduled change in power injection occurs during normal operation. In either case, it is required that the system trajectory, driven by feedback control $(p_\generator^c(\omega), p_\load(\omega))$, converges to a \emph{desired} equilibrium point. In this section, we formalize the criteria for desired equilibrium points by formulating an optimization problem called \emph{optimal frequency control} (OFC), and use OFC to guide the design of control. In Section \ref{sec:stability} we will study the stability of the closed-loop system with the proposed control.   

\subsection{Optimal Frequency Control Problem}
\label{subsec:ofc}

Our objective for the design of feedback control is that, at any closed-loop equilibrium point, $(p^*, d^*)$ is the solution of the following OFC problem, where $d_j^* = D_j \omega_j^*$ are the deviations in power dissipations on buses $j \in \node$. 

\noindent
\textbf{OFC:}
\IEEEbq
\min_{p, d}  & \quad&
	\sum_{j \in \node} \left(  c_j (p_j) + \frac{1}{2 D_j} d_j^2 \right)
\label{eq:ofc.1}
\\
\text{subject to}    &  \quad&\sum_{j \in \node} (  p_j - d_j  )  =  0
\label{eq:ofc.2}
\\
 &  \quad & \underline p_j \leq p_j \leq \overline p_j,    \quad\quad \forall j \in \node.
\label{eq:ofc.3}
\IEEEeq
The term $c_j(p_j)$ in objective function \eqref{eq:ofc.1} is generation cost (if $j \in \generator$) or disutility to users for participating in control (if $j \in \load$). The term $\frac{1}{2D_j} d_j^2$ is a cost for the frequency-dependent deviation in power dissipation, and hence implicitly penalizes frequency deviation at equilibrium. More detail for the motivation of \eqref{eq:ofc.1}, such as why the weighting factor of the second term is set as $\frac{1}{2D_j}$, can be found in \cite{ZhaoTopcuLiLow2014}. The constraint \eqref{eq:ofc.2} ensures power balance over the entire network, and \eqref{eq:ofc.3} are specified bounds on power injections.

We assume the following condition throughout the paper.
\begin{Condition}\label{cond.1}
OFC is feasible. The cost functions $c_j$ are strictly 
convex and twice continuously differentiable on $[ \underline p_j, \overline p_j]$. 
\end{Condition}
\begin{Remark}\label{remark.1}
A load $-p_j$ on bus $j \in \load$ results in a user utility $u_j(-p_j)$, and hence the disutility function can be defined as $c_j(p_j)= -u_j(-p_j)$. In the literature on demand response \cite{LiChenLow2011DR} and economic dispatch \cite{KianiAnnaswamy2012, DorflerPorcoBullo2014}, the load disutility functions or generation cost functions usually satisfy Condition \ref{cond.1}, and in many cases are quadratic functions. See \cite[Sec. III-A]{ZhaoTopcuLiLow2014} for more examples. 
\end{Remark}

\subsection{Design of Decentralized Frequency Feedback Control}
\label{subsec:design}

We use OFC to guide our design of the frequency feedback control. We now specify our design and then prove that any resulting closed-loop equilibrium point is the solution of the OFC problem. Similar to \cite{ZhaoTopcuLiLow2014}, we design $(p_\generator^c(\omega), p_\load(\omega))$ as
\begin{IEEEeqnarray}{rclr}\label{eq:decentralizedcontrol_1}
p_j^ c (\omega_j)&=& \left[(c_j')^{-1}(-\omega_j)\right]_{\underline p_j}^{\overline p_j} &\quad\quad\forall j \in \generator\\
p_j (\omega_j) &=&  \left[(c_j')^{-1}(-\omega_j)\right]_{\underline p_j}^{\overline p_j}  &\quad\quad\forall j \in \load. \label{eq:decentralizedcontrol_2}
\IEEEeq
The function $(c_j')^{-1}(\cdot)$ is well defined due to Condition \ref{cond.1}.
Note that this control is completely decentralized in that for every generator and load, the control decision is a function of frequency deviation measured at its local bus. Only its own cost function and bound need to be known, and no explicit communication or information of system parameters is required. The following theorem shows that this design of control fulfills our objective proposed above.
\begin{Theorem}\label{thm.1}
For any equilibrium $(\theta^*, \omega^*, a_\generator^*,  p_\generator^{c,*}, p^*)$ of the dynamic system \eqref{eq:model:freq}--\eqref{eq:model:turbine} with feedback control \eqref{eq:decentralizedcontrol_1}\eqref{eq:decentralizedcontrol_2}, $(p^*, d^*)$ is the \emph{unique} solution of OFC, where $d_j^* = D_j \omega_j^*$ for $j \in \node$.
\end{Theorem}
\begin{proof}
Consider the dual of OFC with variables $(\lambda, \mu^{+}, \mu^{-})$ where $\lambda$ is a scalar and $\mu^{+}$ and $\mu^{-}$ are both vectors in $\mathbb{R}^{|\node|}$. Since OFC is a convex problem with differentiable objective and constraint functions, by \cite[Sec. 5.5.3]{Boyd2004Convex}, $(p, d; \lambda, \mu^{+}, \mu^{-})$ is optimal for OFC and its dual if and only if it satisfies the following Karus-Kuhn-Tucker (KKT) conditions:
\begin{IEEEeqnarray}{lr}
c_j'(p_j) + \lambda + \mu_j^+ - \mu_j^- = 0  &\quad \forall j \in \node  \label{eq:kkt_1}
\\
d_j-\lambda D_j= 0  &\quad \forall j \in \node  \label{eq:kkt_2}
\\
\sum_{j \in \node} (  p_j - d_j  )  =  0
& \label{eq:kkt_3}
 \\
\underline p_j \leq p_j \leq \overline p_j   & \quad \forall j \in \node \label{eq:kkt_4}
\\
\mu_j^+ \geq 0,\quad \mu_j^- \geq 0  & \quad \forall j \in \node 
\label{eq:kkt_5}
 \\
 \mu_j^+(p_j - \overline p_j) =  \mu_j^-(\underline p_j - p_j) =0  &  \quad \forall j \in \node \label{eq:kkt_6}
\IEEEeq
where \eqref{eq:kkt_1}\eqref{eq:kkt_2} are stationarity conditions, \eqref{eq:kkt_3}\eqref{eq:kkt_4} are primal feasibility, \eqref{eq:kkt_5} is dual feasibility and \eqref{eq:kkt_6} is complementary slackness. With $(\mu^+, \mu^-)$ eliminated, \eqref{eq:kkt_1}--\eqref{eq:kkt_6} are equivalent to \eqref{eq:kkt_2}\eqref{eq:kkt_3} and the following equation
\begin{IEEEeqnarray}{rClr}
p_j &=&  \left[(c_j')^{-1}(-\lambda)\right]_{\underline p_j}^{\overline p_j} &\quad\quad\forall j \in \node.\label{eq:kkt_simplified_1}
\IEEEeq
Note that \eqref{eq:kkt_2}\eqref{eq:kkt_3}\eqref{eq:kkt_simplified_1} has a unique solution $(p^0,d^0; \lambda^0)$ since, by Condition \ref{cond.1}, the left-hand-side of \eqref{eq:kkt_3} is a strictly decreasing function of $\lambda$. On the other hand, for any closed-loop equilibrium $(\theta^*, \omega^*, a_\generator^*, p_\generator^{c,*}, p^*)$ we have 
\begin{IEEEeqnarray}{Clr}
&d_j^* =  D_j \omega^*&\quad\quad\forall j \in \node\label{eq:eq_condition_1}
\\
&\sum_{j \in \node} (  p_j^* - d_j^*  )  =  0 &\label{eq:eq_condition_2}
\\
&p_j^* =  \left[(c_j')^{-1}(-\omega^*)\right]_{\underline p_j}^{\overline p_j} &\quad\quad\forall j \in \node \label{eq:eq_condition_3}
\IEEEeq
where \eqref{eq:eq_condition_1} is due to \eqref{eq:def:eq:5}, and \eqref{eq:eq_condition_2} results from adding up \eqref{eq:def:eq:3} and eliminating $F_j(\theta)$ for $j \in \node$, while \eqref{eq:eq_condition_3} is a result of \eqref{eq:def:eq:4} and the design of control \eqref{eq:decentralizedcontrol_1}\eqref{eq:decentralizedcontrol_2}. Hence $(p^*, d^*; \omega^*) = (p^0,d^0; \lambda^0)$ is the unique solution of    
\eqref{eq:kkt_2}\eqref{eq:kkt_3}\eqref{eq:kkt_simplified_1} and therefore the unique optimal solution of OFC and its dual (ignoring $\mu^+$, $\mu^-$ which may not be unique at optimal).  
\end{proof}

\begin{Remark}
In general for primary frequency control, $\omega^*$ is not zero, which means the frequency is not restored to the nominal value at equilibrium. Traditionally, recovery to nominal frequency needs centralized generator-side secondary frequency control known as AGC \cite{KianiAnnaswamy2012}. Recently, decentralized frequency recovery schemes based on local feedback control and neighborhood communication have been developed from a simplified model, on generators \cite{Li2014AGC}, or on loads \cite{Mallada2014FPC}. Their stability under the more realistic model in this paper still remains to be investigated.
\end{Remark}

Though Theorem \ref{thm.1} ensures the optimality (in the sense of OFC) of closed-loop equilibrium points, it does not guarantee the existence of any equilibrium point. We assume the following condition for the remaining of this paper, which implies an equilibrium point always exists.
\begin{Condition}\label{cond.2}
For the unique optimal $(p^*, d^*)$ of OFC, the set of power flow equations 
\IEEEbq\label{eq:PFequations}
F_j (\theta) =  p_j^* - d_j^*   &\quad\quad \forall j \in \node 
\IEEEeq
is feasible, i.e., has a solution $\theta^*$.
\end{Condition}

\section{STABILITY OF CLOSED-LOOP EQUILIBRIUM}\label{sec:stability}

In this section, we use Lyapunov method to study the stability of the closed-loop system, and derive a sufficient condition for a closed-loop equilibrium point to be asymptotically stable. Our approach for stability analysis is compositional \cite{topcu2009compositional}, in that we find Lyapunov function candidates for the open-loop network and the generator frequency control loops separately, and add them up to form a Lyapunov function.

Theorem \ref{thm.1} implies that for a given system \eqref{eq:model:freq}--\eqref{eq:model:turbine} with all parameters and feedback control \eqref{eq:decentralizedcontrol_1}\eqref{eq:decentralizedcontrol_2} specified, the component $(\omega^*, a_\generator^*, p_\generator^{c,*}, p^*)$ is unique across all equilibrium points. However, this is not the case for $\theta^*$. In general, given $(p^*, d^*)$, \eqref{eq:PFequations} has more than one solutions (if it is feasible), even if we ignore any difference by multiple of $2 \pi$ and regard two solutions as the same one if they have the same difference in $\theta_j$ for all $j \in \node$ \cite{Varaiya1981}. Due to the existence of multiple equilibria, we focus on their \emph{local} asymptotic stability. In particular, we study whether a closed-loop equilibrium is (locally) asymptotically stable if the corresponding open-loop equilibrium, defined below, is asymptotically stable. 

\begin{Definition}\label{def:open-loop-equilibrium}
For an equilibrium $(\theta^*, \omega^*, a_\generator^*, p_\generator^{c,*}, p^*)$ of the closed-loop system, $(\theta^*, \omega^*)$ is an equilibrium for the dynamic system \eqref{eq:model:freq}--\eqref{eq:model:PF} with constant input $p^*$, and is called the corresponding \emph{open-loop equilibrium.}
\end{Definition}

We see that by ``open-loop'' we not only fix the input but also ignore the generator governor and turbine dynamics \eqref{eq:model:governor}\eqref{eq:model:turbine}. We take the same energy function as in \cite{BergenHill1981}, which was used for stability analysis of the open-loop system, as part of the Lyapunov function candidate for the closed-loop system. Given a closed-loop equilibrium $(\theta^*, \omega^*, a_\generator^*, p_\generator^{c,*}, p^*)$, the energy function is 
\IEEEbq
\mathcal V_0 &=& \frac{1}{2}\sum_{j\in \generator} M_j (\omega_j - \omega_j^*)^2  \nonumber
\\
&~&+ \sum_{(i,j) \in \mathcal{E}} \int_{\theta_{ij}^*}^{\theta_{ij}} Y_{ij}(\sin u - \sin \theta_{ij}^* ) du   \label{eq:energyfunction_0}
\IEEEeq    
where $\theta_{ij} := \theta_i - \theta_j$. The requirement for the open-loop equilibrium $(\theta^*, \omega^*)$ to be asymptotically stable imposes a constraint on the range of $\theta_{ij}^*$, so that the integral term in \eqref{eq:energyfunction_0} is positive definite with origin at $\theta_{ij}^*$, for $\theta_{ij}$ in a neighborhood of $\theta_{ij}^*$ \cite{BergenHill1981}. In other words, $\mathcal V_0$ is a valid candidate for Lyapunov function only if the open-loop equilibrium is asymptotically stable. In practice an open-loop equilibrium $(\theta^*, \omega^*)$ typically satisfies $-\frac{\pi}{2}  < \theta_{ij}^* < \frac{\pi}{2}$ for all $(i,j) \in \mathcal E$, which implies it is asymptotically stable. Due to this reason $-\frac{\pi}{2}  < \theta_{ij}^* < \frac{\pi}{2}$ is often considered as a security constraint for power flow problems \cite{DorflerPorcoBullo2014}.

 For simplicity define $ \tilde \omega:= \omega - \omega^*$ and use similar notations for deviations of other variables from their equilibrium values. Taking time derivative of \eqref{eq:energyfunction_0} along the trajectory of $(\theta,\omega)$, we have 
\IEEEbq
\dot{\mathcal{V}}_0 &=& \sum_{j\in \generator}  M_j \tilde \omega_j \dot \omega_j  + \sum_{(i,j) \in \mathcal{E}}  Y_{ij} (\sin \theta_{ij} - \sin\theta_{ij}^* ) (\omega_i - \omega_j) \nonumber
\\  
&=&\sum_{j\in \generator} \tilde \omega_j  (p_j - D_j\omega_j - F_j(\theta)) \nonumber \\
&~& +  \sum_{j\in \node}  \tilde \omega_j (F_j(\theta) - F_j(\theta^*))   \label{eq:energyfunction_dot_0:1}
\\
& = & \sum_{j\in \node}  \tilde \omega_j  (p_j - D_j \omega_j - F_j(\theta^*))  \label{eq:energyfunction_dot_0:2}
\\
& = & - \sum_{j\in \node} D_j\tilde \omega_j^2 + \sum_{j\in \node} \tilde \omega_j \tilde p_j \label{eq:energyfunction_dot_0:3}
\\
& \leq &  - \sum_{j\in \load} D_j\tilde \omega_j^2 + \sum_{j\in \generator} (-D_j \tilde \omega_j^2 + \tilde \omega_j \tilde p_j ) \label{eq:energyfunction_dot_0:4}
\IEEEeq 
where \eqref{eq:energyfunction_dot_0:1} results from \eqref{eq:model:swing}\eqref{eq:model:PF}, and \eqref{eq:energyfunction_dot_0:2} holds since $p_j - D_j\omega_j - F_j(\theta) = 0 $ for $j \in \load$. We get \eqref{eq:energyfunction_dot_0:3} by replacing $F_j(\theta^*)$ in \eqref{eq:energyfunction_dot_0:2} with $p_j^* - D_j \omega_j^*$ (by \eqref{eq:def:eq:3}), and get \eqref{eq:energyfunction_dot_0:4} since
\IEEEbq \nonumber
\tilde \omega_j \tilde p_j = (\omega_j- \omega_j^*)(p_j(\omega_j) - p_j(\omega_j^*)) \leq 0   \quad\quad \forall  j \in \load 
\IEEEeq
where $p_j(\omega_j)$ is a non-increasing function of $\omega_j$ due to its design  \eqref{eq:decentralizedcontrol_2}.

We now study the other parts of Lyapunov function candidate, corresponding to the governor and turbine systems of generators $j \in \generator$. By moving the origin of \eqref{eq:model:governor}\eqref{eq:model:turbine} to the closed-loop equilibrium, we have
\IEEEbq\label{eq:turbine_governor:compact}
\dot {\tilde{y}}_j = A_j \tilde y_j + B_j \tilde p^c_j \quad\quad \forall j \in \generator
\IEEEeq
where $\tilde y_j := [\tilde a_j, \tilde p_j ]^T = [a_j-a_j^*, p_j - p_j^* ]^T$, and $\tilde p^c_j:= p_j^c(\omega_j) - p_j^c(\omega_j^*)$, and 
\IEEEbq\nonumber
A_j := 
\begin{bmatrix}
-\frac{1}{\tau_{g,j}}& 0\\
\frac{1}{\tau_{b,j}}& -\frac{1}{\tau_{b,j}}
\end{bmatrix}
& \quad\quad&
B_j :=
\begin{bmatrix}
\frac{1}{\tau_{g,j}}\\
0
\end{bmatrix}.
\IEEEeq
A classical Lyapunov function for linear system  \eqref{eq:turbine_governor:compact} takes the form
\IEEEbq\nonumber
\mathcal V_j =\frac{1}{2} \tilde y_j^T P_j \tilde y_j
\IEEEeq
where $P_j$ is a positive definite matrix. Hence its time derivative along the system trajectory is
\IEEEbq
\dot{\mathcal V_j} &=& \tilde y_j^T P_j \dot{\tilde y}_j  \nonumber 
\\
&=  & \frac{1}{2} \tilde y_j^T (P_jA_j + A_j^T P_j) \tilde y_j + \tilde y_j^T P_j B_j\tilde p_j^c. \label{eq:Lyapunov_derivative:governorturbine}
\IEEEeq
Since both eigenvalues of $A_j$ are negative, Lyapunov theory tells us that we can find $P_j>0$ such that $P_j A_j + A_j^T P_j <0$. Indeed, if for all $j \in \generator$ we can bound $\dot{\mathcal V_j}$ in a particular form as given in the lemma below, then we can show that the closed-loop equilibrium is asymptotically stable.       
\begin{Lemma}\label{lemma.1}
Suppose, for all $j \in \generator$, that there are $P_j >0 $ such that \eqref{eq:Lyapunov_derivative:governorturbine} implies 
\IEEEbq\label{eq:lemma_1.1}
\dot{\mathcal V}_j \leq -\alpha_j \tilde p_j^2 + \beta_j \tilde \omega_j^2 - \gamma_j (\tilde a_j + \eta_j \tilde p_j)^2 
\IEEEeq
where $\alpha_j, \gamma_j>0$, $\beta_j < D_j$, $\eta_j$ is arbitrary, and $4\alpha_j (D_j - \beta_j) >1$. Then $\mathcal V_{\text{total}}:= \mathcal V_0 + \sum_{j\in \generator} \mathcal V_j$ is a Lyapunov function which ensures the closed-loop equilibrium is asymptotically stable.    
\end{Lemma}
\begin{proof}
If all the conditions in Lemma \ref{lemma.1} are satisfied, then from \eqref{eq:energyfunction_dot_0:4} we have
\IEEEbq
\dot{\mathcal V}_{\text{total}} &\leq& - \sum_{j\in \load} D_j\tilde \omega_j^2 
 - \sum_{j\in \generator } \gamma_j(\tilde a_j - \eta_j \tilde p_j)^2 \nonumber 
\\
& ~& +  \sum_{j\in \generator} (-(D_j-\beta_j) \tilde \omega_j^2 + \tilde \omega_j \tilde p_j - \alpha_j \tilde p_j^2) \nonumber
\IEEEeq
where the third summation is non-positive and evaluates to zero only when $\tilde \omega_j = \tilde p_j = 0$ for all $j \in \generator$. It is straightforward that $\mathcal V_{\text{total}}$ is positive definite and $\dot{\mathcal V}_{\text{total}}$ is negative definite in a neighborhood of the closed-loop equilibrium, with both their origins at the closed-loop equilibrium, and hence the closed-loop equilibrium is asymptotically stable.
\end{proof}
Lemma \ref{lemma.1} inspires us to find a sufficient condition for the closed-loop equilibrium to be asymptotically stable. We drop the subscript $j$ and look at one generator only. 

The choice of $P$ for \eqref{eq:Lyapunov_derivative:governorturbine} to render a bound as in \eqref{eq:lemma_1.1} is not unique. Here we choose $P$ to be diagonal with positive entries $P_{11}$ and $P_{22}$. To ensure $PA + A^T P <0$, we have
\IEEEbq\nonumber
\frac{P_{11}}{\tau_g} > \frac{P_{22}}{4\tau_b}.
\IEEEeq
A calculation from \eqref{eq:Lyapunov_derivative:governorturbine} gives us
\begin{IEEEeqnarray}{rCl} \nonumber
\dot{\mathcal V} &=& -\frac{P_{11}}{\tau_g} \tilde \alpha^2 - \frac{P_{22}}{\tau_b} \tilde p^2 +\frac{P_{22}}{\tau_b} \tilde \alpha \tilde p+ \frac{P_{11}}{\tau_g} \tilde \alpha \tilde p^c
\\
&=&  -(\frac{P_{22}}{\tau_b}- \frac{P_{22}}{4\gamma \tau_b^2})\tilde p^2 
  + \frac{P_{11}^2}{4 \tau_g (P_{11} - \gamma \tau_g)} (\tilde p^c)^2  \nonumber
\\
&~& - \gamma \left(\tilde \alpha -\frac{P_{22}}{2\gamma \tau_b}  \tilde p \right)^2 \nonumber  
\\
&~& - (\frac{P_{11}}{\tau_g} - \gamma)\left(\tilde \alpha - \frac{P_{11} \cdot \tilde p^c}{2(P_{11} - \gamma \tau_g)}\right)^2 \label{eq:der:stability_condition:2} 
\IEEEeq
for arbitrary $\gamma \in (\frac{P_{22}}{4 \tau_b}, \frac{P_{11}}{\tau_g})$. Motivated by the form of \eqref{eq:lemma_1.1}, we assume the following condition to bound the term of $(\tilde p^c)^2$ in \eqref{eq:der:stability_condition:2} by a term of $\tilde \omega^2$.    
\begin{Condition}\label{cond.3}
The generator control $p^c(\omega)$, designed as \eqref{eq:decentralizedcontrol_1}, is locally Lipschitz continuous in a neighborhood of $\omega^*$, i.e., for all $\omega$ in that neighborhood, there exists a real constant $L \geq 0$ such that 
\IEEEbq\nonumber
|p^c(\omega) - p^c(\omega^*)| \leq L |\omega- \omega^*|.
\IEEEeq
\end{Condition}
With Condition \ref{cond.3} we have $(\tilde p^c)^2 \leq L^2 \tilde \omega^2$ in the neighborhood referred to, and hence \eqref{eq:der:stability_condition:2} leads to \eqref{eq:lemma_1.1} by taking
\IEEEbq\nonumber
\alpha_j = \frac{P_{22}}{\tau_b}- \frac{P_{22}}{4\gamma \tau_b^2} &\quad\quad& \beta_j =  \frac{P_{11}^2L^2}{4 \tau_g (P_{11} - \gamma \tau_g)}
\\
\gamma_j = \gamma  &\quad\quad&  \eta_j  = -\frac{P_{22}}{2\gamma \tau_b}
\IEEEeq
which satisfies $\alpha_j > 0$,  $\gamma_j > 0$. We make the following transformation of variables 
\IEEEbq\label{eq:transform_to_xi}
\xi = \frac{P_{22}}{4\tau_b}, \quad \sigma = \frac{\xi}{\gamma }, \quad z =\frac{\tau_g \gamma}{P_{11}} 
\IEEEeq
so that $\xi > 0$, $0 <\sigma < 1$ and $0 < z <1$. Hence the conditions $\beta_j < D_j$ and $4\alpha_j (D_j - \beta_j) >1$ in Lemma \ref{lemma.1} can be written as 
\IEEEbq\label{eq:der:stability_condition:3.1}
&D - \frac{L^2 \xi}{4 \sigma z(1-z)} > 0&  
\\
&16 \xi (1-\sigma)\left(D-\frac{L^2 \xi}{4 \sigma z(1-z)}\right) >1&. \label{eq:der:stability_condition:3.2}
\IEEEeq
Subject to the domain of $(\xi,\sigma,z)$ and \eqref{eq:der:stability_condition:3.1}, the maximum of the left-hand-side of \eqref{eq:der:stability_condition:3.2} is $\frac{D^2}{L^2}$, achieved by successively taking $z = \frac{1}{2}$, $\frac{L^2 \xi}{\sigma} = \frac{D}{2}$, and $\sigma = \frac{1}{2}$. Hence as long as $L<D$, we can find positive definite, diagonal matrix $P$ through inverse transformation of \eqref{eq:transform_to_xi}, such that all the conditions in Lemma \ref{lemma.1} are satisfied. 

We come back to the problem for the entire network and recover the indexing subscripts. Our result is formally summarized in the following theorem. We skip the proof since it is obvious given all the preceding analysis.    
\begin{Theorem}\label{thm.2}
Consider an equilibrium $(\theta^*, \omega^*, a_\generator^*, p_\generator^{c,*}, p^*)$ for the closed-loop system. Suppose the corresponding open-loop equilibrium $(\theta^*, \omega^*)$ is asymptotically stable. Suppose, for all $j \in \generator$, that Condition \ref{cond.3} holds within a neighborhood of $\omega^*$ for a constant $ 0 \leq L_j < D_j$. Then the closed-loop equilibrium $(\theta^*, \omega^*, a_\generator^*, p_\generator^{c,*}, p^*)$ is asymptotically stable.
\end{Theorem}

Note that Theorem \ref{thm.2} only gives a sufficient condition for stability. Its derivation depends on the particular (diagonal) structure of $P$ and particular form of bound on the derivative of Lyapunov function as \eqref{eq:lemma_1.1}. Hence this condition may be conservative. It is possible in practice that the control is designed so that $L_j > D_j$ but the system is still stable. As ongoing work, we try to obtain less conservative stability conditions by searching for $P$ in a larger subset of positive definite matrices.

\section{CASE STUDY}\label{sec:simulation}

We illustrate the performance of the proposed control through a simulation of the IEEE New England test system shown in Fig. \ref{fig:IEEE_39}.

\begin{figure}[thpb]
      \centering
      \includegraphics[height=7.5 cm]{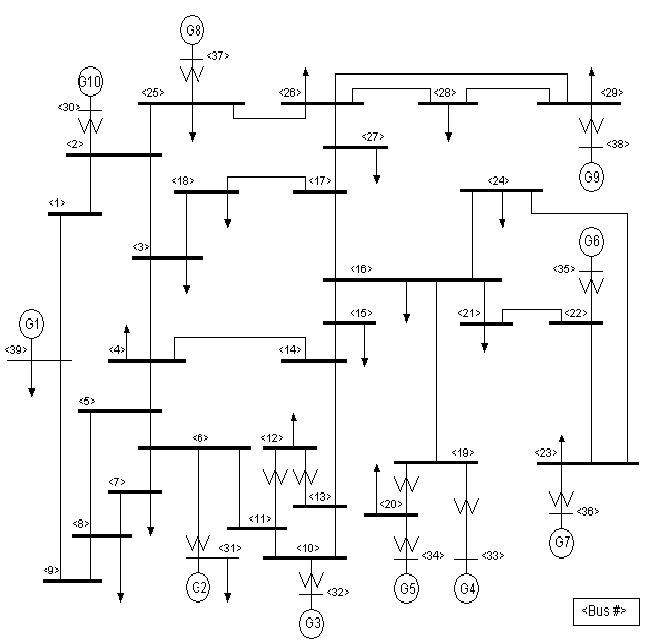}
      \caption{IEEE New England test system \cite{cheung2009power}.}
      \label{fig:IEEE_39}
   \end{figure} 
This system has 10 generators and 39 buses, and a total load of about 60 per unit (pu) where 1 pu represents 100 MVA. Details about this system including parameter values can be found in Power System Toolbox \cite{cheung2009power}, which we use to run the simulation in this section. Compared to the model \eqref{eq:model:swing}--\eqref{eq:model:governor}, the simulation model is more detailed and realistic, with transient generator dynamics, excitation and flux decay dynamics, changes in voltage and reactive power, non-zero line resistances, etc. 

The primary frequency control of generator or load $j$ is designed with cost function $c_j (p_j) = \frac{R_j}{2}(p_j - p_j^{\text{set}})^2$, where $p_j^{\text{set}}$ is the power injection at the setpoint, an initial equilibrium point solved from static power flow problem. By choosing this cost function, we try to minimize the deviation of system state from the setpoint, and have the control $p_j = \left[p_j^{\text{set}} - \frac{1}{R_j} \omega_j\right]_{\underline p_j}^{\overline p_j}$ from \eqref{eq:decentralizedcontrol_1}\eqref{eq:decentralizedcontrol_2} \footnote{Only the control for $p_j$ is written since the same holds for $p_j^c$}. We consider the following two cases in which the generators and loads have different control capabilities and hence different $[\underline p_j, \overline p_j]$: 
\begin{enumerate}
\item All the 10 generators have $[\underline p_j, \overline p_j] = [ p_j^{\text{set}}(1-c), p_j^{\text{set}}(1+c)]$, and all the loads are uncontrollable; 
\item Generators 2, 4, 6, 8, 10 (which happen to provide half of the total generation) have the same bounds as in case (1). All the other generators are uncontrollable, and all the loads have $[\underline p_j, \overline p_j] = [ p_j^{\text{set}}(1 + c/2),  p_j^{\text{set}}(1 - c/2)]$, if we suppose $p_j^{\text{set}} \leq 0$ for loads $j \in \load$. 
\end{enumerate}
Hence cases (1) and (2) have the same total capacity for control across the network. Case (1) only has generator control while in case (2) the set of generators and the set of loads each has half the total control capacity. We select $c=10\%$, which implies the total capacity for control is about 6 pu. For all $j \in \node$, the feedback gain $1/R_j$ is selected as $25p_j^{\text{set}}$, which is a typical value in practice meaning a frequency change of 0.04 pu (2.4 Hz) causes the change of power injection from zero all the way to the setpoint. Note that this control is the same as frequency droop control, which implies that indeed frequency droop control implicitly solves an OFC problem with cost functions we use here. However, our controller design allows for more flexibility by using other cost functions.    

In the simulation, the system is initially at the setpoint with 60 Hz frequency. At time $t = 0.5$ second, buses 4, 15, 16 each make 1 pu step change in their real power \emph{consumptions}, causing the frequency to drop. Fig. \ref{fig:g_vs_l} shows the frequencies of all the 10 generators under the two cases above, (1) with red and (2) with black. We see in both cases that frequencies of different generators are close to each other with relatively small differences during transient, and they are synchronized towards the new steady state. Compared with generator-only control, the combined generator-and-load control improves both the transient and steady-state frequency, even though the total control capacities in both cases are the same.    
   \begin{figure}[thpb]
      \centering
      \includegraphics[height=4.9 cm]{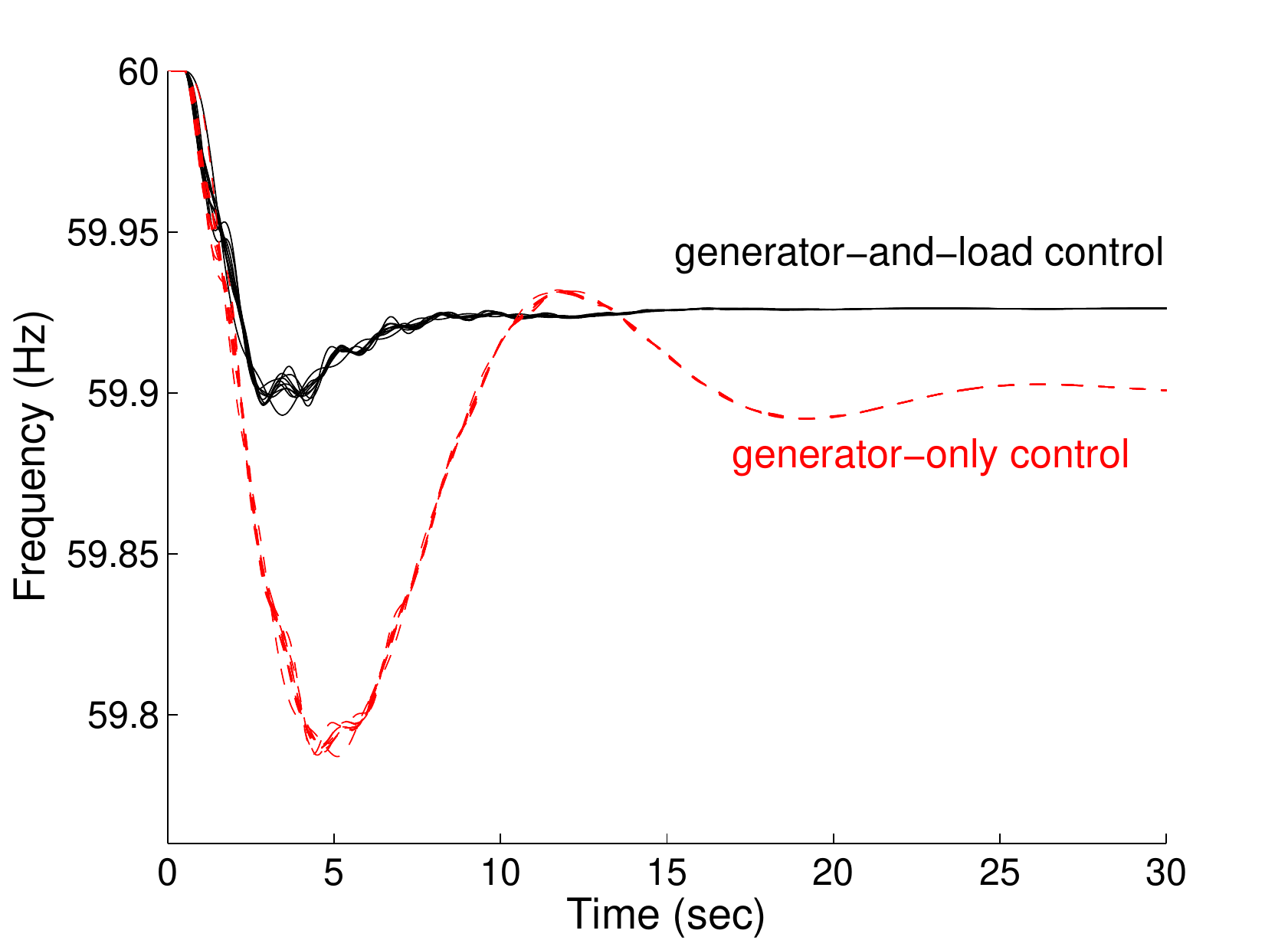}
      \caption{Frequencies of all the 10 generators under case (1) only generators are controlled (red) and case (2) both generators and loads are controlled (black). The total control capacities are the same in these two cases.}
      \label{fig:g_vs_l}
   \end{figure}

\section{CONCLUSION AND FUTURE WORK}\label{sec:conclusion}

We have presented a systematic method to jointly design generator and load-side primary frequency control, by formulating an optimal frequency control (OFC) problem to characterize the desired equilibrium point of the closed-loop system.  OFC minimizes the total generation cost and user disutility subject to power balance across the network. The control is completely decentralized, depending only on local frequency. For the closed-loop system under the proposed control, stability analysis with Lyapunov method has led to a sufficient condition for an equilibrium point to be asymptotically stable. Simulations show that, the combined generator-and-load control improves both transient and steady-state frequency, compared to the traditional control on generator side only, even when the total control capacity remains the same.    

We have shown the stability of equilibrium points without studying their attraction regions, and particularly the change of attraction regions due to closing the loop with the proposed control. It is an interesting topic to work on for the next step. Moreover, as discussed above it is our future work to understand how conservative the sufficient stability condition in Theorem \ref{thm.2} is and derive less conservative, or sufficient and necessary conditions for stability. 
We are also interested in performing stability analysis with a more detailed model like that in \cite{tsolas1985structure} with flux decay and time-varying voltage magnitudes and reactive power flows.





\end{document}